\documentclass[10pt]{cai}

\graphicspath{{figs/}}

\usepackage{algorithm}
\usepackage{algorithmic}
\usepackage{amsthm}
\usepackage{cleveref}
\usepackage[mathscr]{euscript}
\usepackage{adjustbox}
\newtheorem{theorem}{Theorem}
\newtheorem{lemma}{Lemma}
\crefname{lemma}{Lemma}{Lemmas}

\newcommand{\hcoupling}{h^{cou}}
\newcommand{\hcohesion}{h^{coh}}
\newcommand{\hadd}{h^{add}}
\newcommand{\intraedge}{E_{\boxtimes}}
\newcommand{\intraedgemax}{\intraedge^{\max}}
\newcommand{\interedge}{E_{\bowtie}}

\begin{document}

\begin{center}
\vspace{-0.2in}
  \title{Opti Code Pro: A Heuristic Search-based Approach to Code Refactoring}
  \maketitle

  \thispagestyle{empty}

  \begin{tabular}{cc}
    Sourena Khanzadeh\upstairs{\affilone,*}, Samad Alias Nyein Chan\upstairs{\affilone}, Richard Valenzano\upstairs{\affilone}, Manar Alalfi\upstairs{\affilone}
   \\[0.25ex]
   {\small \upstairs{\affilone} Toronto Metropolitan University} \\
  \end{tabular}
  
  \emails{
    \upstairs{*}samad.chan@torontomu.ca 
    \upstairs{*}sourena.khanzadeh@torontomu.ca 
    \upstairs{*}rick.valenzano@torontomu.ca 
    \upstairs{*}manar.alalfi@torontomu.ca 
    }
  \vspace*{-0.2in}
\end{center}
\begin{abstract}
This paper presents an approach that evaluates best-first search methods to code refactoring. The motivation for code refactoring could be to improve the design, structure, or implementation of an existing program without changing its functionality. To solve a very specific problem of coupling and cohesion, we propose using heuristic search-based techniques on an approximation of the full code refactoring problem, to guide the refactoring process toward solutions that have high cohesion and low coupling. We evaluated our approach by providing demonstrative examples of the effectiveness of this approach on random state problems and created a tool to implement the algorithm on Java projects.
\end{abstract}
\begin{keywords}{Keywords:}
Code Refactoring, Heuristic Search, A*, Software Engineering
\end{keywords}
\copyrightnotice

\section{Introduction}
\label{intro}
Refactoring code is a strategy utilized in software engineering and programming that involves changing the internal organization of a pre-existing code without altering its functionality. During refactoring, the software's capabilities remain intact while the design, arrangement, and/or execution are improved \cite{fowler2018refactoring}. The significance of code refactoring is to help sustain and enhance the maintenance and overall quality of software systems over time. Over time, as new functionalities are introduced, the codebase can become intricate and challenging to comprehend and preserve. Refactoring code tackles these challenges by enhancing the design, arrangement, and execution of existing code, resulting in a more legible, maintainable, and scalable software system. 

This paper aims to introduce \textit{Opti Code Pro}, a new tool that employs best-first search algorithms to automate and streamline the laborious process of deciding code re-factoring changes for an object-oriented system. Specifically, this system focuses on minimizing a program's \emph{coupling} and maximizing \emph{cohesion} by changing the dependencies between classes and modules. In object-oriented software development, cohesion and coupling deem as important principles that play a crucial role in the development of a well-organized, sustainable codebase. Cohesion refers to the degree to which the elements within a module are related to one another, while coupling refers to the degree to which one module depends on another. Low coupling and high cohesion are generally considered to be desirable properties of well-designed programs \cite{CouplingCohesionProblem}.

Deploying the principles of coupling and cohesion to suggest code refactoring changes are the major focus of this tool. Our objective is to introduce an approximation of the Code Refactoring Domain, demonstrate a unique heuristic search-based approach to enhance coupling and cohesion, and implement it using the Opti Code Pro tool in actual software development projects - the source code for the tool, heuristic techniques and experiments in this paper will be included upon publication.

By employing these search-based techniques, we aim to enhance the design, structure, and implementation of existing programs while maintaining their functionality. 
We experiment on randomly generated program architectures to better understand the performance of our system and demonstrate that it can be used to find effective refactoring on 2 open-source Java projects.

\section{Search Background}
A \emph{search problem} can be defined as the tuple $\mathscr{S} = <S, s_0, \mathcal{G}, A, \beta, k>$, where $S$ is a finite set of \emph{states}, $s_0 \in S$ is the \emph{initial state}, $\mathcal{G}: S \rightarrow \lbrace true, false \rbrace$ is a \emph{goal test function} that returns true if and only if the given state is a goal, $A$ is a set of \emph{actions}, $\beta(a, s)$ is the \emph{state transition function}, and $k(a, s) \geq 0$ stands for a \emph{cost} of doing an action $a$ on state $s$. We use $succ(s)$ to denote the successors of state $s$, meaning the set of states that can be reached from $s$ using a single action.

The objective of a search problem is to find a sequence of actions $a_0, ..., a_n$, called a \emph{solution plan}, that can be applied starting in $s_0$, such that the result is a state $s_g$ such that $G(s)=true$. The cost of any plan is given by the sum of the costs of the actions along that plan. A solution plan is \emph{optimal} if its cost is minimal with respect to all solutions, and $w$-\emph{admissible} for some weight $w \geq 1$, if the cost of that solution is no more than $w$ times larger than the optimal solution.

Heuristic search algorithms use a \emph{heuristic function} to guide a search for a solution plan. A heuristic function $h: S \rightarrow \mathbb{R}^{\geq 0}$ helps prioritize which states to explore during the search. This is typically done by estimating the cost to reach the nearest goal state from the current state. A heuristic $h$ is  \emph{admissible} if it never overestimates the optimal cost to reach a goal state from any state in $S$. A heuristic $h$ is \emph{consistent} if $h(c) \geq h(p) - k(a, p)$ for every pair of states $p$ and $c$ where $c \in succ(p)$ and $\beta(a, p)=c$. It is well-known that if $h$ is consistent and $h(s)=0$ for any goal state $s$, then $h$ is necessarily admissible.

The most well-known heuristic search algorithm is the \emph{A*} algorithm. A* is a best-first search algorithm that iteratively expands and stores a set of partial plans starting from the initial state until one is found that reaches the goal. For any state that has been found during the search, this algorithm maintains the cost of the plan found to that state, denoted as $g(s)$. On every iteration, where $h$ is a heuristic function, A* selects the state with the lowest value of $f(s) = g(s) + h(s)$ for expansion, from amongst all states that have previously been found but not expanded. The successors of this state are then used to extend the path to $s$, and these are added to the set of states that have been reached but not expanded. A* is guaranteed to find optimal paths if the given heuristic is admissible.

Another common best-first search algorithm is \emph{Weighted A* (WA*)} \cite{POHL1970193}.
This method only varies in the way they prioritize states for expansion. Specifically, WA* uses the evaluation function $f_w(s) = g(s) = wh(n)$ where $w \geq 1$ is a user-defined parameter. 
WA* is guaranteed to only return solutions that are $w$-admissible if $h$ is admissible and typically finds solutions faster than A*. 
\section{Class-Level Refactoring} \label{sec:refactoring_definition}
Our focus is on class-level and module-level dependency refactoring of an object-oriented software project.
For this purpose, we approximate a project as having three main components: \textbf{modules}, \textbf{classes}, and \textbf{dependencies}. Modules represent units of functionality that independently contain one or more classes. Classes are used to create objects and encapsulate data and behavior regarding the objects. Dependencies represent the relationship between classes, where a class may be reliant on another to fulfill its functionality. For example, a method invocation from a class used within another.
\begin{figure*}[t!]
\centering
\includegraphics[width=0.3\textwidth]{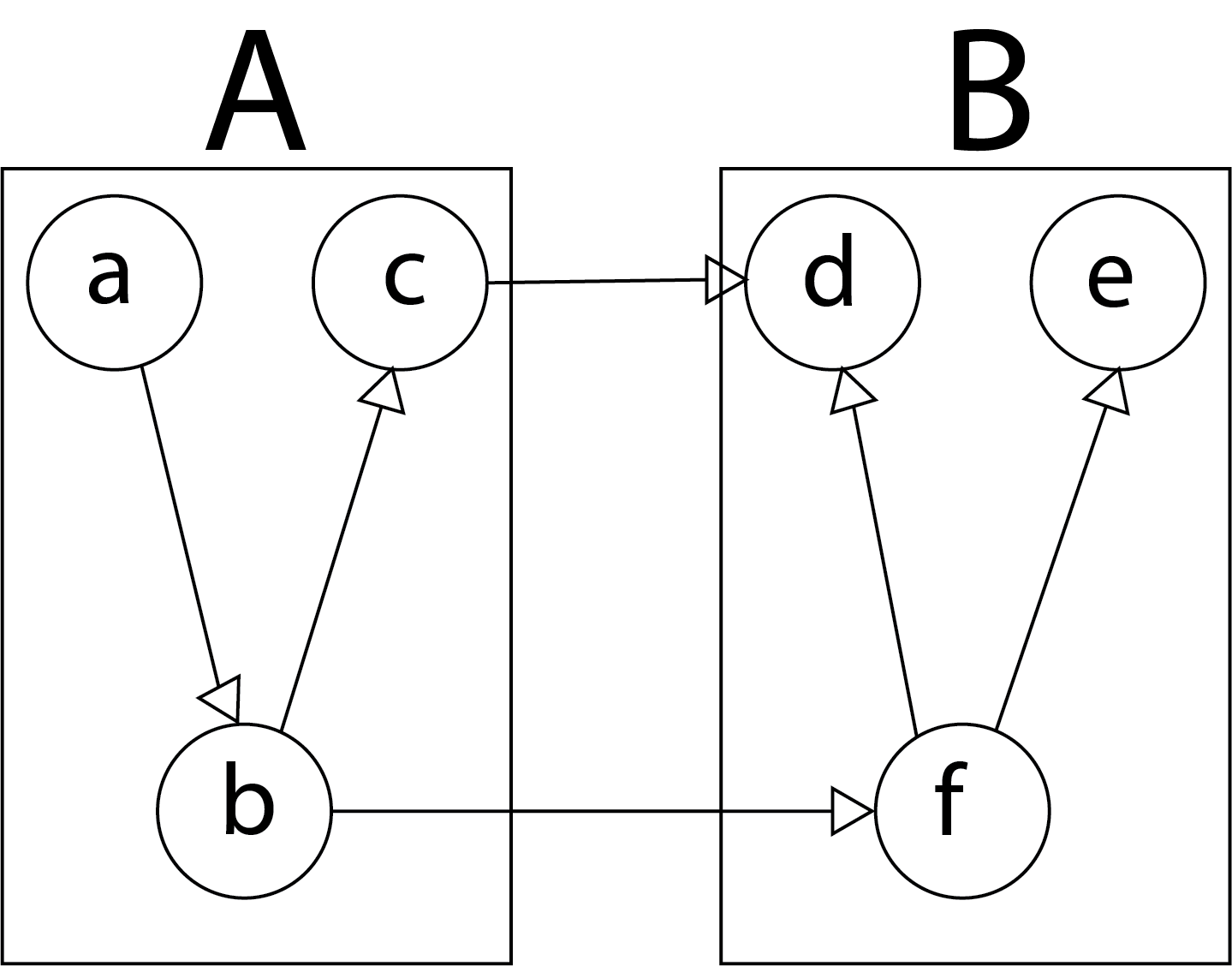} 
\caption{A running example of Modules, Classes, and Dependencies}
\vspace{-0.1in}
\label{fig:example_state}
\end{figure*}
Figure \ref{fig:example_state} shows a simple project with two modules, $A$ and $B$, each containing three classes. The dependencies between these classes are indicated by arrows, with the arrow pointing from the dependent class to the class that it depends on. Clearly, we can represent the project as a graph where each vertex corresponds to a class, each edge corresponds to a dependency, and the classes are partitioned in mutually exclusive sets such that there is one for each module. Formally, we represent a project as a tuple $s = \langle V, E, \mathcal{M} \rangle$ where $V$ is a set of vertices, $E$ is a directed set of edges over $V$, and $\mathcal{M} = \lbrace M_1, ..., M_k \rbrace$ is a set of modules, where each $M_i \subseteq V$ is the set of vertices in module $M_i$, and there is no overlap between modules (\textit{ie.} $M_i \cap M_j = \emptyset$ for every $M_i, M_j \in \mathcal{M}$). Note that it is possible for a state to contain both edge $(x, y)$ and edge $(y, x)$.

It is important to note that dependencies are transitive, so if we have edges $(x, y)$ and $(y, z)$, then $x$ also implicitly depends on $z$. However, this can also make edges superfluous, and thus possible to delete. For example, an edge $(b, d)$ would be unnecessary in Figure \ref{fig:example_state} due to edges $(b,c)$ and $(c, d)$. Thus, we will say that a class $x$ depends on $y$ if edge $(x, y)$ is in the state or there is a directed path of edges from $x$ to $y$ in the state.

Edges can also be partitioned into two types: \emph{intra-module edges} and \emph{inter-module edges}. 
An intra-module edge is between vertices in the same module: $e = (v_1, v_2)$ where $v_1, v_2 \in M_i$. An inter-module edge is an edge between vertices in different modules: $e = (v_1, v_2)$ where $v_1 \in M_i$ and $v_2 \in M_j$ where $M_i \neq M_j$. 
We will also use the following notation for a state $s$: $\intraedge^{M}(s)$ to denote the set of intra-edges between vertices in $M$, $\interedge^{M_1,M_2}(s)$ to denote the set of all inter-edges between $M_1$ and $M_2$, $\intraedge(s)$ to denote the set of intra-edges in all modules, and $\interedge(s)$ to represent the set of all inter-edges.

This notation also allows us to formally measure cohesion and coupling. Specifically, we measure cohesion as the total number of intra-edges in the current project, namely $\intraedge(s)$. Coupling is then measured as the total number of inter-edges, namely $\interedge(s)$.

Our objective is now to improve the cohesion and coupling of a software project by adding or removing edges.
However, the final product of the refactoring is to maintain the original dependencies of the project.
That is, the final outcome may add or remove edges, but the dependencies must be preserved.
\section{Code Refactoring as a Search Problem}
In this section, we will describe our approach to finding refactoring. At a high level, we use best-first search on a simplified version of the problem to find possible solutions to the refactoring problem. The method can be run multiple times with different hyperparameter settings to generate solutions with different cohesion levels. A simple repair method is also used to increase the likelihood that the generated solutions maintain the original dependencies.
As a running example, we will use the simple project in Figure \ref{fig:example_state}.
\subsection{The Best-First Search Action Set} \label{sec:action_set}
As stated in Section \ref{sec:refactoring_definition}, our problem states are given by a tuple which uses a graph representation of the problem along with a partitioning of the vertices into modules.
A simple action set would then be to allow one action for adding each missing edge in the graph and removing each existing edge in the graph. This would correspond to 10 possible add actions and 6 possible delete edges in our running example.

Unfortunately, early tests suggested this greatly slowed down the generation of solutions, so we made several simplifications.
First, we do not allow adding inter-edges since doing so would not change cohesion and would hurt/increase coupling. Similarly, we do not allow for the deletion of intra-edges, since that would hurt/decrease cohesion without changing coupling.
This leaves only 8 add actions and 2 delete actions in the project in our running example.

As this was still leading to a slow search, we further pruned the set of actions applicable so that only a single add action and a single delete action were applicable in each state.
The edge in question was selected arbitrarily. 
We will discuss the consequences of these choices below in Section \ref{sec:repair}. For simplicity, we set the cost of all actions to be 1 in our current system. These values may be updated to account for notions of the desirability of adding/deleting an edge, or the difficulty of performing that refactoring step. However, we leave such extensions as future work.
\subsection{Goal States and Cohesion Aggression}
To capture the objective of improving cohesion and coupling, we use these measures to define the goal states.
In particular, a state $s$ is a goal state if and only if the following two conditions hold:
\begin{enumerate}
    \item There is no more than 1 inter-edge between any two modules. Formally, this means that for any two modules $M_1$ and $M_2$, $\interedge^{M_1,M_2}(s) \leq 1$.
    \item For a given hyperparameter $\alpha$ where $0 \leq \alpha \leq 1$, the cohesion of $s$ must be at least a factor of $\alpha$ of the maximum possible cohesion. Formally, this means that $\intraedge(s) \geq \alpha \sum_{M \in \mathcal{M}} M(M-1)$.  
\end{enumerate}
Note that $\intraedgemax(s) = \sum_{M \in \mathcal{M}} M(M-1)$ 
is used to denote the maximum number of intra-edges possible in $s$.

The hyperparameter $\alpha$ is known as the \emph{cohesion aggression}. Its main purpose is to determine the desired level of cohesion within a module. 
When set to 1, Opti Code Pro establishes a bidirectional connection between all classes within the same module, resulting in the highest degree of cohesion. A value of 0 will not require any additional dependencies be added. The value of $\alpha$ should thus be selected based on the specific goals and constraints of the project, with a default value of 0.5 commonly used. The appropriate tuning of $\alpha$ by Opti Code Pro enables a balance between improving the cohesion of the module and preserving its existing structure and functionality.

Notice that the goal test does not require that the final solution actually satisfy the original dependencies of the problem.
This is because doing so was found to be a very expensive check at each step of the search.
Doing so may also make it difficult to find satisfying solutions when combined with the requirement that there is at most one inter-edge between every pair of modules.
We will address this issue in Section \ref{sec:repair}, but note that the aggression level will also affect the likelihood that the final solution is valid.
For example, when the aggression is equal to 1.0, it is relatively easy to remove inter-edges while maintaining the original dependencies, since the additional intra-edges substantially increase the number of superfluous inter-edges.
\subsection{Completeness, Optimality, and Solution Repair} \label{sec:repair}
As stated above, by simplifying the problem and not checking if the original dependencies of a problem are maintained in a goal state, we are not guaranteed that our search returns valid solutions to our main problem. Similarly, by pruning all but two actions in every state, we are not guaranteed to find the shortest solution, even to the problem of reducing coupling to only one inter-edge per pair of modules and reaching a cohesion level required by the given cohesion aggression parameter.

We do note though that when A* is used with an admissible heuristic for this simplification, we are finding optimal solutions to this simplified problem. As our results show below, we are also generally able to find valid solutions that improve both cohesion and coupling for reasonable cohesion aggression levels. This is especially true when using the solution repair method we use which will be described below.

To see why we still get effective solutions, we first note that even though we are only allowing two actions per state, the fact that action sequencing does not affect the resulting state means that delaying when an action is applicable will not necessarily affect solution quality. For example, notice that if we add edge $(a, c)$ to the example state in Figure \ref{fig:example_state} prior to edge $(d, e)$, the resulting state will be the same if we add $(d, e)$ prior to $(a, c)$.
Thus, if we need to add $(a, c)$ and $(d, e)$ for our solution, the fact that adding $(d, e)$ may not be applicable in the current state will not necessarily affect solution quality if it is applicable later on.

Of course, if the arbitrarily selected set of actions is not actually needed in the best solution, the search may have to apply them in order to reach states where the necessary actions are applicable. This can lead to plans that cost more than what would be found if the full action set were allowed in every state. Having said that, if the cohesion aggression is high enough, it is likely that this effect is minor since most of the added edges will have to be included anyways.

Regarding the validity of the final solutions, we can also use a simple \emph{solution repair} approach to significantly increase the likelihood that we end up with solutions that maintain the original problem's dependencies.
Specifically, we found that the majority of cases where we were returning invalid cases were due to the fact that if there were two-way dependencies between modules, our goal test would force one of them to be deleted. For example, suppose that our example in Figure \ref{fig:example_state} also had an edge $(f, b)$. The best-first search solution would necessarily delete two of the three edges $(c, d)$, $(b, f)$, and $(f, b)$. Suppose that it deleted $(c, d)$ and $(f, b)$. Then, in this case, all dependencies from classes in module A to those in module B would be lost.

To remedy this, we use the following approach. We iteratively add the reverse of every inter-edge remaining in our solution, but only if it improves validity. For example, in the above example, we would add the edge $(f, b)$ since the remaining edge was $(b, f)$. After each edge is added, we check the original dependencies. If the resulting state is now valid, we return that solution. Otherwise, we proceed to the next inter-edge and add its reverse edge. This continues until a valid solution is found or we run out of inter-edges.
As we will see in the experiments, this simple approach will generally result in more valid solutions, especially for higher aggression levels.

In summary, while our method is not guaranteed to find optimal or valid solutions to the class refactoring problem, it is still likely to be finding good solutions that we will show are generally valid. It will also be able to do so very quickly, as we will see in the experiments. Thus, Opti Code Pro can be seen as using an incomplete approach that is effective in practice in a similar vein as other Artificial Intelligence approaches such as the use of Tabu search \cite{glover1998tabu} for combinatorial optimization problems or Las Vegas algorithms \cite{luby1993optimal} for various decision problems.
\subsection{Implemented on 2 Different Java Projects}
The proposed code refactoring technique using best-first search algorithms was tested on three Java projects scraped from GitHub. The projects were extracted using regular expressions to gather information on classes, modules, and dependencies. The implementation aimed to evaluate the technique's effectiveness and usefulness in real-world projects by measuring improvements in cohesion, coupling, and other metrics and comparing the results to other code refactoring methods. The results showed the potential of best-first search algorithms for code refactoring, as we were able to identify low cohesion and high coupling areas and suggest refactoring actions that improved the code's quality and maintainability.
\section{Heuristics for Refactoring}
In this section, we describe the heuristic functions we use in our best-first search approach. We also prove that these heuristics are consistent and admissible. Thus, any solution found with A* when guided by these heuristics will be optimal. We also note that these properties will hold regardless of whether the action set uses the simplified action sets used in Opti Code Pro or the full action set of the main problem as described in Section \ref{sec:action_set}.
\subsection{The Coupling  Heuristic}
Recall that we measure the coupling of a program as the number of inter-module edges (\textit{ie.} the number of dependencies between classes in different modules) and that our goal test also requires that there be no more than one inter-module dependency between any two modules.
Thus, our \emph{coupling heuristic} estimates the minimum number of coupling actions that must be taken to reach such a goal state.
Intuitively, this is given by counting the number of ``extra'' edges between every pair of modules, meaning beyond the single allowed by the goal test.
Formally, this heuristic is given as follows:

\begin{align} \label{coupling_heuristic}
    \hcoupling(s) = \sum_{M_1,M_2 \in \mathcal{M}} \max(0, |\interedge^{M_1,M_2}(s)| - 1) 
\end{align}
For example, the value of the coupling heuristic in our example in Figure \ref{fig:example_state} is 1.

Notice that for each pair of modules, our heuristic takes the maximum with 0. This is because if a pair of modules already have 0 or 1 dependency between them, there is no need to remove a dependency.
This is also important to show that the heuristic is consistent and admissible.
To do so, we first use the following lemma

\begin{lemma}\label{lemma:coupling_edge_change}
    If $c \in succ(p)$ and is achieved with action $a$, the following hold:
\begin{enumerate}
    \item If $a$ adds an intra-edge, then $\hcoupling(p) = \hcoupling(c)$.
    \item If $a$ removes an inter-edge, then $\hcoupling(c) \geq \hcoupling(p) - 1$.
\end{enumerate}
\end{lemma}
\begin{proof}
If $a$ adds an intra-edge, it has no effect on the number of inter-edges. Thus, the value of equation \ref{coupling_heuristic} will be the same before and after the action and so $\hcoupling(p) = \hcoupling(c)$.

Suppose $a$ removes an inter-edge between modules $M_1$ and $M_2$. If $|\interedge^{M_1,M_2}(p)| > 1$, then the heuristic value will decrease by 1 from $p$ to $c$. If $|E^{m_1,m_2}_{intra}(p)| = 1$, then the heuristic value will not change. In either case, $\hcoupling(c) \geq \hcoupling(p) - 1$.
\end{proof}

This lemma allows us to prove the consistency and admissibility of $\hcoupling$:
\begin{theorem} \label{theorem:coupling_admissible}
    $\hcoupling$ is consistent and admissible.
\end{theorem}
\begin{proof}
    Since all actions have a cost of 1, to prove $\hcoupling$ we must show that $\hcoupling(c) \geq \hcoupling(p) - 1$ for any pair of states $p$ and $c$ where $c$ is a successor of state $p$. This follows immediately from Lemma \ref{lemma:coupling_edge_change}.

    If $s$ is a goal state, then there is at most 1 edge between any pair of modules by definition. When this holds, clearly $\hcoupling(s) = 0$ by the definition in Equation \ref{coupling_heuristic}. Therefore, $\hcoupling$ is admissible since it is consistent and $\hcoupling(s) = 0$ for any goal state $s$.
\end{proof}
\subsection{Cohesion Heuristic}
For our second heuristic, recall that our goal test checks if the total number of intra-edges is at least $\alpha \intraedgemax$. 
Thus, our \emph{cohesion heuristic} measures how many intra-module edges must be added before satisfying this goal test the current state:
\begin{align}
    \hcohesion(s) = \max(0, \lceil \alpha \intraedgemax \rceil - \sum_{M \in \mathcal{M}} |\intraedge^M(s)|
\end{align}

Intuitively, this heuristic will encourage the search to add intra-edges until the $\alpha E^{max}_{intra}$ threshold is reached.
Taking the maximum with 0 ensures that once this threshold is reached, the heuristic remains 0.

To show this heuristic is consistent and admissible, we use the following lemma:
\begin{lemma}\label{lemma:cohesion_edge_change}
    If $c \in succ(p)$ and is achieved with action $a$, the following hold:
\begin{enumerate}
    \item If $a$ adds an intra-edge, then $\hcohesion(c) \geq \hcohesion(p) - 1$.
    \item If $a$ removes an inter-edge, then $\hcohesion(p) = \hcohesion(c)$.
\end{enumerate}
\end{lemma}
\begin{proof}
    If $a$ adds an intra-edge, there are two cases to consider. If $\hcohesion(s) = 0$, adding an intra-edge will not affect the heuristic further, so $\hcohesion(p) = \hcohesion(c)$. Otherwise, the total number of intra-edges will get one closer to the $\alpha \intraedgemax$ threshold and so $\hcohesion(c) = \hcohesion(p) - 1$. Thus, $\hcohesion(c) \geq \hcohesion(p) - 1$.

    If $a$ deletes an inter-edge, then the number of intra-edges stays constant from $p$ to $c$. As such, $\hcohesion(p) = \hcohesion(c)$.
\end{proof}

This lemma now immediately implies that $\hcohesion$ is consistent and admissible:
\begin{theorem} \label{theorem:cohesion_admissible}
    $\hcohesion$ is consistent and admissible.
\end{theorem}
\begin{proof}
    This follows an analogous argument as Theorem \ref{coupling_heuristic}. The consistency of $\hcohesion$ follows from Lemma \ref{lemma:cohesion_edge_change}, which implies admissibility since $\hcohesion(s) =  0$ for any goal state $s$.
\end{proof}
\subsection{The Additive Heuristic}
Since each of $\hcoupling$ and $\hcohesion$ focuses on a different part condition of the goal test, a natural extension is to take their sum as an estimate of the effort needed to simultaneously satisfy both conditions.
Below, we show that this \emph{additive} heuristic, $\hadd = \hcoupling + \hcohesion$, is also consistent and admissible.
Intuitively, this holds because each counts a mutually disjoint set of actions that need to be performed to achieve the goal.
We formalize this below:
\begin{theorem}
    The heuristic $\hadd = \hcoupling + \hcohesion$ is consistent and admissible.
\end{theorem}
\begin{proof}
    Consider any two pairs of nodes $p$ and $c$, where $c \in succ(p)$, achieved with action $a$. If $a$ removes an inter-edge, then $\hcoupling(c) \geq \hcoupling(p) - 1$ by Lemma \ref{lemma:coupling_edge_change} and Lemma \ref{lemma:cohesion_edge_change} guarantees that $\hcohesion(p) = \hcohesion(c)$. Together, these imply that $\hcoupling(c) + \hcohesion(c) \geq \hcoupling(p) + \hcohesion(p) - 1$.

    The case where $a$ adds an intra-edge is analogous. Lemma \ref{lemma:coupling_edge_change} ensures that $\hcoupling(p) = \hcoupling(c)$ and Lemma \ref{lemma:cohesion_edge_change} ensures that $\hcohesion(c) \geq \hcohesion(p) - 1$. Together, these imply that $\hcoupling(c) + \hcohesion(c) \geq \hcoupling(p) + \hcohesion(p) - 1$.

    Since $\hadd(c) \geq \hadd(p) - 1$ holds in both cases, $\hadd$ is consistent. Since Theorems \ref{theorem:coupling_admissible} and \ref{theorem:cohesion_admissible} show that $\hcoupling(s_g) = \hcohesion(s_g) = 0$ for any goal state $s_g$, clearly $\hadd(s_g) = \hcoupling(s_g) + \hcohesion(s_g) = 0$.
    Together with the fact that $\hadd$ is consistent, this proves it is admissible.
\end{proof}

Notice that for any state $s$, $\hadd(s) \geq \hcohesion(s)$ and $\hadd(s) \geq \hcoupling(s)$.
Since, $\hadd$ therefore \emph{dominates} the other heuristics and all three heuristics are consistent, we are guaranteed that $\hadd$ will expand no more nodes than $\hcohesion$ and $\hcoupling$ (aside from the effects of tie-breaking) \cite{dechter:generalized}.
Our empirical results are consistent with this fact.
The runtime of $\hadd$ will also be improved, demonstrating that the improved accuracy is more than compensating for the extra overhead of calculating two heuristics instead of one.
\section{Experiments}
In this section, we empirically test our approach to class-level refactoring problems. We begin by describing our experimental setup and our use of randomly generated project structures for testing. Using these problems, we will then perform a comprehensive evaluation comparing the relative performance of the different heuristics and the value of using repair for different levels of aggressiveness.
We will then test different suboptimal heuristic search algorithms on this problem to see if they effectively trade-off solution quality for speed. Finally, we show experiments where we applied Opti Code Pro on two real-world Java projects. 
\subsection{Experimental Setup}
For testing our approach, we generated 100 randomly directed graphs with vertices and partitioned the vertices into modules.
These are intended to represent project architectures.
All problems have 25 classes partitioned randomly between 15 modules. 
To generate dependencies, we uniformly sample a number $x$ from $1$ to $25\cdot 24$ (\textit{ie.} the maximum number of edges), and select a random subset of $x$ possible edges. These are then added to the graph.
We then test different combinations of algorithm, heuristic, and aggression on Opti Code.
These experiments are all run on a Mac Studio 2022 with an M1 Max chip with a RAM of 32GB and Mac OS Ventura 13.1.
\subsection{Evaluating Heuristics and Repair Using A*}
In this section, we focus on evaluating our system when using A* to better understand the impact of the heuristics and our repair approach.
To do so, recall that since our heuristics are admissible, we are guaranteed to find optimal solutions to the restricted problem where we do not test for validity in the goal test and only have 2 actions applicable per state, but this is not necessarily optimal or complete to the full refactoring problem.
\begin{table}
\centering

\resizebox{1\textwidth}{!}{%
\begin{tabular}{|l|rrr|rrr|rrr|rrr|}
\toprule
& \multicolumn{3}{c|}{Aggression of 0.25} & \multicolumn{3}{c|}{Aggression of 0.5} & \multicolumn{3}{c|}{Aggression of 0.75} & \multicolumn{3}{c|}{Aggression of 1.0} \\
     Heuristic &        Exp. &     Time (s) &  Cost &       Exp. &     Time (s) &  Cost &        Exp. &     Time (s) &  Cost &     Exp. &     Time (s) &  Cost \\
\midrule
Zero &           44485 &  2.194 &  97.61 &          45511 &  2.208 &  98.74 &           47003 &  2.263 &  101.17 &        47793 &  2.266 &  104.47 \\
Coupling &            6069 &  0.203 &  97.61 &           9439 &  0.256 &  98.74 &           20859 &  0.669 &  101.17 &        43290 &  2.417 &  104.47 \\
Cohesion &           44475 &  2.270 &  97.61 &          45473 &  2.285 &  98.74 &           46812 &  2.424 &  101.17 &        47412 &  2.587 &  104.47 \\
Additive &            5811 &  0.210 &  97.61 &           5905 &  0.214 &  98.74 &            6763 &  0.281 &  101.17 &         8600 &  0.598 &  104.47 \\
\bottomrule
\end{tabular}
 }
 \caption{Performance of A* with Different Heuristics Under Varying Levels of Aggression}
\label{table:atarperformance}
\end{table}

\begin{figure}[t!]
\includegraphics[width=0.8\linewidth]{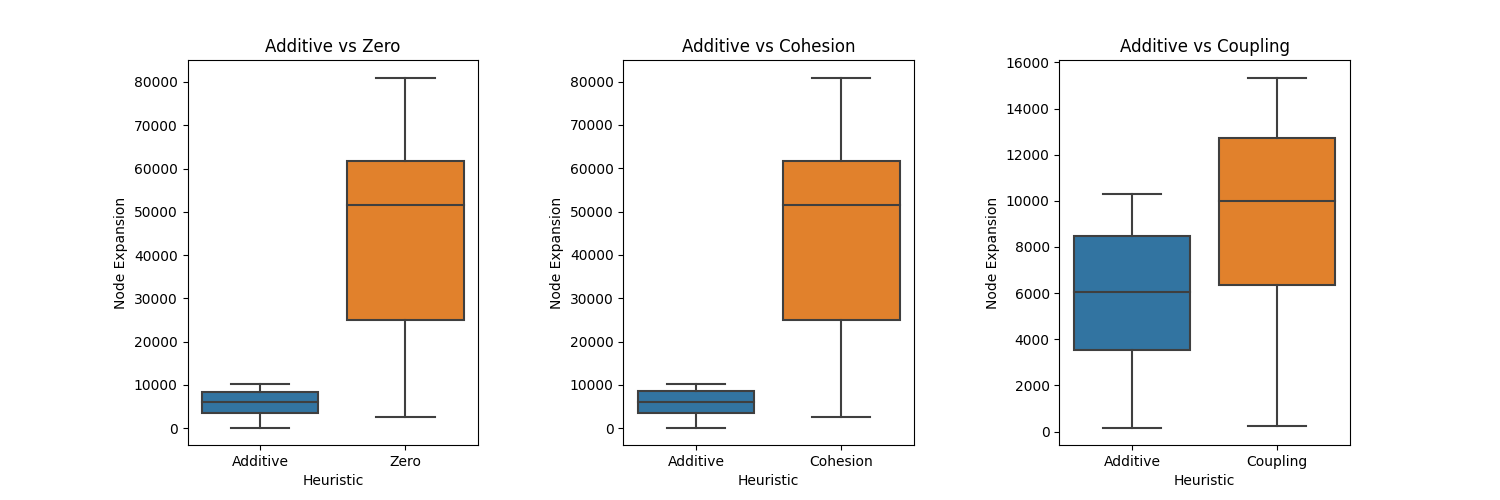}
\vspace{-0.2in}
\caption{Comparing the distribution of node expansions for different heuristics.}
\vspace{-0.15in}
\label{figure2:node_expansions_additive}
\end{figure}

The results for this comparison are shown in Table \ref{table:atarperformance}.
The table shows the average number of expansions needed to find a solution (labeled \emph{Exp.}), the average time in seconds, and the average cost of the solutions found.
In addition to our three proposed heuristics, we also include data for the \emph{zero heuristic}, which always returns a heuristic estimate of 0 for every state.
This is intended as a baseline to help demonstrate the value of the constructed heuristics.

In all cases, we see that using our approximation of the refactoring problem, we are able to find solutions very quickly.
We note that for the lowest aggression value, $0.25$, cohesion provides minimal improvement over the zero heuristic, and the additive heuristic provides minimal improvement over coupling. 
This is because in many cases, the initial state is already satisfying the aggression criteria in our goal test or requires minimal intra-edge add actions to get there.
Thus, the guidance provided by the coupling heuristic becomes more critical.
For high aggression values, the guidance provided by cohesion does become more valuable. 

As stated above, $\hadd$ will dominate the other heuristics so it is unsurprising that it outperforms them in terms of node expansions in all cases.
This benefit is minimal over the coupling heuristic for lower aggression values, and so the added accuracy does not always compensate for the added time to also compute the cohesion heuristic as needed by $\hadd$.
However, the results are still very close for low aggression values and the benefit becomes significant for high aggression values, which motivates the use of $\hadd$ in all cases.
Figure \ref{figure2:node_expansions_additive} provides another view of this behavior in the form of a box-and-whisker plot showing the distribution of node expansions when using the different heuristics when the aggression is 1.0.

Next, we evaluate the validity of solutions found by our approach, and consider the value of the solution repair method.
Table \ref{table:astar_validity} shows a comparison of the validity of the solutions found by A* before and after the repair method.
The table shows the percentage of problems on which valid solutions have been found, and the average decrease in the overall number of inter-edges in the solution found by A* on the 4 aggression levels (0.25, 0.5, 0.75, and 1).
First, we notice that increasing the aggression makes it significantly easier for our approach to finding more valid solutions.
As a result, repairing the solution is critical for low aggression values in which a high percentage of the solutions found do not satisfy the original dependencies.
Repair does fix most, but not all of these issues for low aggression values, and is able to find satisfying solutions to all problems with a high aggression level. 
Also, note that as the aggression values increase, the number of inter-edges added back by the repair method (\textit{ie.} the difference between edges deleted before and after repair) decreases. Hence, the repair method has to add back fewer edges as the aggression values increase.
More importantly though, in general, even after the repair, the average number of inter-edges decreased compared to the original solution does not change much, meaning that we are still effectively reducing coupling compared to the original solution.

\begin{table}
\centering
\small
\begin{tabular}{|c|cc|cc|}
\toprule
     & \multicolumn{2}{c|}{Before Repair} & \multicolumn{2}{c|}{After Repair}\\
     Aggression & Validity & Edges Deleted & Validity &  Edges Deleted\\
\midrule
0.25 & 0.49& 96.28& 0.90 & 94.74\\
0.50 & 0.52& 96.28& 0.94 & 95.28\\
0.75 & 0.60& 96.28& 0.97 & 95.70\\
1.00 & 0.79& 96.28& 1.00 & 95.96\\
\bottomrule
\end{tabular}

\caption{The impact of repair on solution validity and inter-edge deletions.}
\label{table:astar_validity}
\end{table}


\subsection{Using Weighted A* to Find Refactorings}
\begin{table}
\centering
\resizebox{1\textwidth}{!}{%
\begin{tabular}{|l|rrrr|rrrr|rrrr|rrrr|}
\toprule
& \multicolumn{4}{c|}{Aggression of 0.25} & \multicolumn{4}{c|}{Aggression of 0.5} & \multicolumn{4}{c|}{Aggression of 0.75} & \multicolumn{4}{c|}{Aggression of 1.0} \\
     &        Exp. &     Time (s) &  Cost &  Valid &       Exp. &     Time (s) &  Cost &  Valid &        Exp. &     Time (s) &  Cost &  Valid &     Exp. &     Time (s) &  Cost &  Valid \\
   %
\midrule
A* &            5811 &  0.210 &   97.61 &  0.95 &           5905 &  0.214 &   98.74 &  0.96 &            6763 &  0.281 &  101.17 &  0.98 &         8600 &  0.598 &  104.47 &   1.0 \\
WA*(5) &            5572 &  0.201 &   97.67 &  0.95 &           5632 &  0.205 &   98.80 &  0.96 &            5767 &  0.240 &  101.23 &  0.98 &         5891 &  0.410 &  104.53 &   1.0 \\
WA*(10) &            5572 &  0.202 &   97.67 &  0.95 &           5632 &  0.207 &   98.80 &  0.96 &            5767 &  0.239 &  101.23 &  0.98 &         5891 &  0.405 &  104.53 &   1.0 \\
WA*(25) &            5572 &  0.200 &   97.67 &  0.95 &           5632 &  0.206 &   98.80 &  0.96 &            5767 &  0.238 &  101.23 &  0.98 &         5891 &  0.405 &  104.53 &   1.0 \\
WA*(50) &            5572 &  0.201 &   97.67 &  0.95 &           5632 &  0.205 &   98.80 &  0.96 &            5767 &  0.239 &  101.23 &  0.98 &         5891 &  0.403 &  104.53 &   1.0 \\
WA*(100) &            5572 &  0.200 &   97.67 &  0.95 &           5632 &  0.209 &   98.80 &  0.96 &            5767 &  0.235 &  101.23 &  0.98 &         5891 &  0.409 &  104.53 &   1.0 \\
\bottomrule
\end{tabular}
 }

 
 \caption{The performance of WA* as part of Opti Code Pro.}
\label{table:suboptimal_algorithms}
\end{table}


In this section, we test Opti Code Pro when using WA* instead of A* to see if the use of these algorithms further speeds up the time needed to find solutions. Since \ref{sec:repair} shows additive heuristic and repairing to give better results, we use them in the following experiments.

Table \ref{table:suboptimal_algorithms} shows the average number of expansions (labeled Exp.), the average time in seconds, the average cost of the solutions found, and the validity of the suboptimal algorithms across 4 aggression levels.

First, notice that for all weights (5, 10, 25, 50, and 100), WA* gives the same results. 
Also, the average cost of the solutions found is approximately the same for A* and WA*. 



We also note that compared to A*, WA* reduces the number of node expansions and time for all aggression values. Specifically, as the aggression increases, the improvement in the number of node expansions using WA* increases. This is appropriate as with lower aggression, the cohesion heuristic would get satisfied early in the algorithm run, and the heuristic would be based only on coupling. However, as the aggression increases, the additive heuristic will involve the cohesion heuristic more and the increased strength of the heuristic reduces the number of nodes expanded by WA*.

Consequently, similar to node expansions, the time taken by WA* also reduces more as the aggression increases.
However, the fact that all weights yield the same results suggests that performance is plateauing quickly.


 


\subsection{Tests on Open Source Projects}
In this section, we evaluate Opti Code Pro on two open-source projects.
The projects were selected to represent different sized real-world projects.
Details on the projects can be found in Table \ref{tab:projectsInfo}.

\begin{table}[htbp]
    \centering
    \small
    \begin{tabular}{|c|c|c|c|}
    \toprule
         Project & GitHub Link & Num of Modules & Num of Classes \\
    \midrule
         1 & \url{https://github.com/iluwatar/java-design-patterns} & 76 & 173\\
         2 & \url{https://github.com/iluwatar/uml-reverse-mapper} & 6 & 12 \\
    \bottomrule
    \end{tabular}
    \caption{Statistics on the open-source projects used for testing.}
    \label{tab:projectsInfo}
\end{table}

To apply our approach to these projects, we developed a tool to scrape all publicly available open-source Java projects of a given GitHub user.
The classes and modules of the project are then extracted and outputted to a JSON file.
Finally, the JSON file is then given as input to Opti Code Pro, which converts it to an initial state for our heuristic search-based engine.

The performance of our system on two projects when using different aggression values is shown in Table \ref{tab:open_source}.
For these experiments, we used the additive heuristic and the solution repair approach.
The columns show the number of inter-edge deletions, the number of intra-edge additions, the node expansions, runtime, and whether the solution found is valid (1) or not (0) after repair.

The table shows that the number of deletions (inter-edges deleted) is the same for the different aggression values, but the number of additions (intra-edges added) increases with higher aggression values.
Again, our methods are performing better for higher aggression values.
In particular, we note that we had to use a much higher cohesion to reach valid solutions for the larger project. 
This suggests that it is necessary to develop ways to better capture missing dependencies during the search if the tool is to be used for lower aggression values.
However, the tool is clearly effective at finding valid solutions with high cohesion, and it does so in a reasonable amount of time even for very large software projects.
\begin{table}
\centering
\resizebox{1\textwidth}{!}{%
\begin{tabular}{|cc|rrrrr|rrrrr|rrrrr|}
\toprule
       &       & \multicolumn{5}{c|}{Aggression of 0.85} & \multicolumn{5}{c|}{Aggression of 0.95} & \multicolumn{5}{c|}{Aggression of 1.0} \\
      Proj. &   Alg.    & Dels & Adds & Exp. &    Time (s) & Valid. & Dels & Adds & Exp. &    Time (s) & Valid. &  Dels & Adds & Exp. &    Time (s) & Valid. \\
\midrule
1 & A* &              23 &   347 &       370 &  7.1817 &     0 &              23 &   396 &       789 &   9.8946 &     1 &           23 &   420 &      1232 &  11.4687 &     1 \\
2       & A* &               0 &     7 &      7 &  0.0011 &     1 &               0 &     9 &      9 &   0.0006 &     1 &            0 &     9 &      9 &   0.0006 &     1 \\
\bottomrule
\end{tabular}
}
\caption{Performance of Different Java Projects Under Varying Levels of Aggression}
\vspace{- 0.6 cm}
\label{tab:open_source}
\end{table}
\vspace{- 0.2 cm}
\section{Related Work}
\vspace{- 0.1 cm}
Heuristics and Search-based solutions have been briefly mentioned, in theory, to solve a variety of problems like requirement prioritization, finding a good design, test data collection, and re-factoring \cite{searchbasedsegeneral}.

One approach to refactoring is to pre-define rules, also known as ``code smells", which are indicative of potential problems in code. For instance, Fowler \cite{fowler2018refactoring} describes a set of code smells and provides a list of refactoring that can address each code smell. However, this approach is dependent on a predetermined set of rules and is not applicable to all situations.

To address this limitation, Ouni et al. \cite{ouni2016multi} proposed a multi-objective search approach to detecting refactoring. Multiple objectives, such as reduction of design errors and design maintenance of design semantics, and using a search algorithm to determine the best sequence of refactoring to enhance code quality while minimizing the consequence of existing code, allowing for more contextually relevant refactoring recommendations. Ouni et al. \cite{ouni2016multi} tackles low-level code refactoring like moving methods, fields, and classes. However, this paper tackles dependency refactoring to specifically improve cohesion and coupling, so work done by Ouni et al. \cite{ouni2016multi} is incomparable.

Greedy and A*-based algorithms have also been utilized to identify refactoring solutions. For instance, Wongpiang and Muenchaisri \cite{wongpiang2014comparing} compare the efficacy of different best-first search algorithms for identifying refactoring that maximizes software maintainability. However, Wongpiang and Muenchaisri \cite{wongpiang2014comparing} only discuss fixing internal changes to fix Long Method, Large Class, and Feature Envy bad smells. 

Hayash et al. \cite{searchbasedcoderefactoring} proposed a search-based approach to detecting concurrent refactoring. This approach considers changes to the structure of the program as transitioning between states and utilizes a search algorithm to discern the sequence of refactoring that was performed. Hayash et al. \cite{searchbasedcoderefactoring} only considers fine-grain refactoring like extracting methods, moving methods, and removing parameters from methods.

Code refactoring consists of many code modifications, complicating the matter. Instead of solving all code refactoring techniques more generally, we simplify this by approximating the problem and focusing on the problem of cohesion and coupling.

\section{Conclusion and Future Work}
In this work, we formulate a class-level refactoring problem as a search problem, then use heuristic search algorithms applied to an approximation of the problem to generate refactorings with improved coupling and cohesion. We also developed several heuristics for use by these algorithms. By guiding the refactoring process towards solutions that have high cohesion and low coupling, the proposed method can improve the design, structure, and implementation of a program without changing its functionality. The experiments on randomly generated problems, and the implementation on open-source Java projects demonstrate the effectiveness of this approach, and show that it is a promising solution for addressing the common problem of coupling and cohesion in software development.

In future work, coupling aggression could be added as a hyperparameter to lower the aggressiveness on the number of inter-edge required by the goal test function. Ideally, we would also be able to directly find solutions for the full refactoring problem, and this will require the development of safe action pruning techniques. The problem could also be made more realistic by considering different cost functions that capture the difficulty of different refactoring actions. Better feasibility evaluation approaches could also be incorporated. This could include systems like Ref-Finder \cite{RefFinder}, a software used to evaluate changes between 2 versions (before and after suggested refactoring from the algorithm), or other architecture evaluation techniques \cite{ouni2016multi}.

This method also represents a first step towards more general automated refactoring tools. This could include proposing a refactoring technique to use, such as suggesting where to split a large module to create smaller inter-connected modules that each have higher cohesion. Possibly, an AI system could even implement the suggestions. 

\printbibliography[heading=subbibintoc]

@book{fowler2018refactoring,
  title={Refactoring: improving the design of existing code},
  author={Fowler, Martin},
  year={2018},
  publisher={Addison-Wesley Professional}
}

@article{ouni2016multi,
  title={Multi-criteria code refactoring using search-based software engineering: An industrial case study},
  author={Ouni, Ali and Kessentini, Marouane and Sahraoui, Houari and Inoue, Katsuro and Deb, Kalyanmoy},
  journal={ACM Transactions on Software Engineering and Methodology (TOSEM)},
  volume={25},
  number={3},
  pages={1--53},
  year={2016},
  publisher={ACM New York, NY, USA}
}

@article{searchbasedcoderefactoring,
  title={Search-based refactoring detection from source code revisions},
  author={Hayashi, Shinpei and Tsuda, Yasuyuki and Saeki, Motoshi},r
  journal={IEICE TRANSACTIONS on Information and Systems},
  volume={93},
  number={4},
  pages={754--762},
  year={2010},
  publisher={The Institute of Electronics, Information and Communication Engineers}
}

@article{searchbasedsegeneral,
title = {Search-based software engineering},
journal = {Information and Software Technology},
volume = {43},
number = {14},
pages = {833-839},
year = {2001},
issn = {0950-5849},
doi = {https://doi.org/10.1016/S0950-5849(01)00189-6},
url = {https://www.sciencedirect.com/science/article/pii/S0950584901001896},
author = {Mark Harman and Bryan F Jones}
}

@article{CouplingCohesionProblem,
author = {Saxena, Vipin and Kumar, Santosh},
year = {2012},
month = {01},
pages = {},
title = {Impact of Coupling and Cohesion in Object-Oriented Technology},
volume = {05},
journal = {Journal of Software Engineering and Applications},
doi = {10.4236/jsea.2012.59079}
}

@inproceedings{wongpiang2014comparing,
  title={Comparing Heuristic Search Methods for Selecting Sequence of Refactoring Techniques Usage for Code Changing},
  author={Wongpiang, Ratapong and Muenchaisri, Pornsiri},
  booktitle={International MultiConference of Engineers and Computer Scientistis (IMECS2014)},
  volume={1},
  year={2014}
}

@article{POHL1970193,
title = {Heuristic search viewed as path finding in a graph},
journal = {Artificial Intelligence},
volume = {1},
number = {3},
pages = {193-204},
year = {1970},
issn = {0004-3702},
doi = {https://doi.org/10.1016/0004-3702(70)90007-X},
url = {https://www.sciencedirect.com/science/article/pii/000437027090007X},
author = {Ira Pohl}
}

@INPROCEEDINGS{RefFinder,  author={Prete, Kyle and Rachatasumrit, Napol and Sudan, Nikita and Kim, Miryung},  booktitle={2010 IEEE International Conference on Software Maintenance},   title={Template-based reconstruction of complex refactorings},   year={2010},  volume={},  number={},  pages={1-10},  doi={10.1109/ICSM.2010.5609577}}

@incollection{glover1998tabu,
  title={Tabu search},
  author={Glover, Fred and Laguna, Manuel},
  booktitle={Handbook of combinatorial optimization},
  pages={2093--2229},
  year={1998},
  publisher={Springer}
}

@inproceedings{luby1993optimal,
  title={Optimal speedup of Las Vegas algorithms},
  author={Luby, Michael and Sinclair, Alistair and Zuckerman, David},
  booktitle={Theory and Computing Systems, 1993., Proceedings of the 2nd Israel Symposium on the},
  pages={128--133},
  year={1993},
  organization={IEEE}
}

@article{dechter:generalized,
  author    = {Rina Dechter and
               Judea Pearl},
  title     = {Generalized Best-First Search Strategies and the Optimality of A*},
  journal   = {J. {ACM}},
  volume    = {32},
  number    = {3},
  pages     = {505--536},
  year      = {1985},
  url       = {https://doi.org/10.1145/3828.3830},
  doi       = {10.1145/3828.3830},
  bibsource = {dblp computer science bibliography, https://dblp.org}
}

\end{document}